\documentclass[letterpaper, 10 pt, conference]{ieeeconf}  

\IEEEoverridecommandlockouts                              

\usepackage[font = footnotesize]{caption}
\usepackage{graphicx,graphics,epsfig,epstopdf,float}
\usepackage{amsmath} 
\usepackage{amssymb}  
\usepackage{mathrsfs}
\usepackage{enumerate}
\usepackage{colortbl}
\usepackage{color}  
\usepackage{bm}
\usepackage{psfrag}
\usepackage{cite}
\usepackage{algorithmic}
\usepackage[ruled,vlined,linesnumbered]{algorithm2e}
\usepackage{mdwlist}
\usepackage{longtable}
\usepackage{array}
\usepackage{acronym}  

\usepackage{bm}
\usepackage{pifont}
\usepackage{graphicx}
\usepackage{bbm}
\usepackage{makecell}
\usepackage{pifont}
\usepackage{stfloats}
\usepackage{geometry}
\usepackage{dsfont}

\usepackage{etoolbox}

\makeatletter
\patchcmd{\@algocf@start}
  {-1.5em}
  {0pt}
  {}{}
\makeatother

\DeclareMathAlphabet{\mathsfbr}{OT1}{cmss}{m}{n}
\SetMathAlphabet{\mathsfbr}{bold}{OT1}{cmss}{bx}{n}
\DeclareRobustCommand{\msf}[1]{%
  \ifcat\noexpand#1\relax\msfgreek{#1}\else\mathsfbr{#1}\fi
}

\makeatletter
\newcommand{\msfgreek}[1]{\csname s\expandafter\@gobble\string#1\endcsname}
\makeatother

\DeclareRobustCommand{\mcal}[1]{%
  \ifcat\noexpand#1\relax\mathnormal{#1}\else\cal{#1}\fi
}
\DeclareRobustCommand{\BM}[1]{%
  \ifcat\noexpand#1\relax\bm{\boldUppercaseItalicGreek{#1}}\else\bm{#1}\fi
}
\makeatletter
\newcommand{\boldUppercaseItalicGreek}[1]{\csname var\expandafter\@gobble\string#1\endcsname}
\makeatother

\definecolor{BLUE}{rgb}{0,0,1}

\newtheorem{theorem}{Theorem}

\newtheorem{definition}{Definition}
\newtheorem{remark}{Remark}
\newtheorem{lemma}{Lemma}

\acrodef{its}[ITS]{intelligent transportation system}

\title{\LARGE \bf
Distributed Scheduling at Non-Signalized Intersections with Mixed Cooperative and Non-Cooperative Vehicles
}

\author{Feihong~Yang and Yuan~Shen 
\thanks{The authors are with the Department of Electronic Engineering, Tsinghua University, and also with the Beijing National Research Center for Information Science and Technology, Beijing 100084, China (e-mail: {yfh17@mails.tsinghua.edu.cn}, {shenyuan\_ee@tsinghua.edu.cn}).}
}

\begin{document}
\bstctlcite{IEEEexample:BSTcontrol}
\newgeometry{top=19.1mm,left=19.1mm,right=19.1mm,bottom=19.1mm}


\maketitle


\begin{abstract}
Intersection management with mixed cooperative and non-cooperative vehicles is crucial in next-generation transportation systems. For fully non-cooperative systems, a minimax scheduling framework was established, while it is inefficient in mixed systems as the benefit of cooperation is not exploited. This letter focuses on the efficient scheduling in mixed systems and proposes a two-stage policy that makes full use of the cooperation relation. Specifically, a long-horizon self-organization policy is first developed to optimize the passing order of cooperative vehicles in a distributed manner, which is proved convergent when inbound roads are sufficiently long. Then a short-horizon trajectory planning policy is proposed to improve the efficiency when an ego-vehicle faces both cooperative and non-cooperative vehicles, and its safety and efficiency are theoretically validated. Furthermore, numerical simulations verify that the proposed policies can effectively reduce the scheduling cost and improve the throughput for cooperative vehicles.
\end{abstract}



\acresetall

\acresetall		

\section{INTRODUCTION}\label{sec:intro}
Traffic intersections are recognized as a critical bottleneck in transportation systems, and safe and efficient intersection management protocols are crucial for the development of modern cities \cite{SchAlbEisLom:B21}. Advancing technologies in autonomous and connected vehicles yielded various types of non-signalized scheduling schemes, which can improve the performance metrics such as delay and throughput by utilizing the intelligence of vehicles \cite{RioMal:J17, BanChaMal:J22, GhoSei:J22, ChaMal:J22}. However, over a long period of time, connected vehicles will have to face a mixed environment with both cooperative and non-cooperative vehicles. Thus, developing scheduling policies while addressing the trajectory uncertainty of non-cooperative vehicles becomes an important challenge.

Some existing works focus on reducing the uncertainty of non-cooperative vehicles through different methods. For example, in \cite{PruZhaDubHuaKis:C19, ChaMan:J22, TiaLiKolYilGir:J22}, mathematical models are adopted to characterize the behavior of non-cooperative vehicles; learning-based policies are introduced in \cite{LinZhaShaWanYuZha:C19, KaiWanCheHaoZhaLiu:C20}, in which the uncertainty of non-cooperative vehicles is limited by the given sample set. However, these methods cannot ensure the safety when the non-cooperative vehicles violate the given behavior models and the sampled patterns. In contrast, some other works such as \cite{MilPerOniGon:J10, HafDel:J11, ColDel:J15} take the full uncertainty into consideration, while they only focus on the safety issue without considering the scheduling efficiency.

\begin{figure} [!t]
  \centering%
  \includegraphics[width=0.38\textwidth]{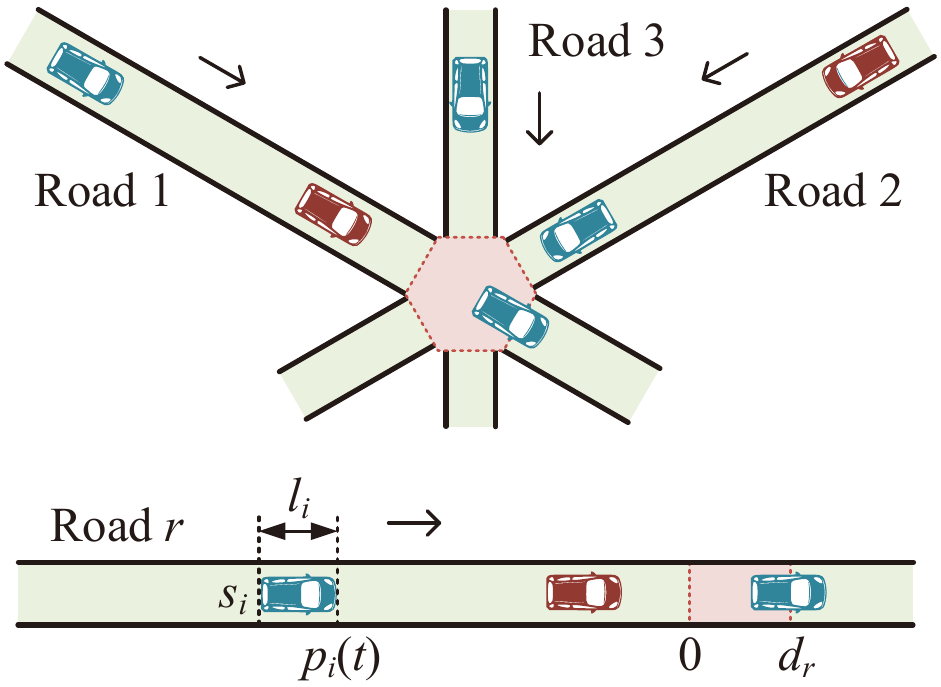}
  \caption{The upper subfigure illustrates the scheduling system with $R=3$, where the blue and red vehicles represent the cooperative and non-cooperative ones, respectively. The lower subfigure shows the physical meanings of system parameters on a general road $r$.}
  \label{fig:1}
\end{figure}

To overcome the challenges caused by non-cooperative vehicles, a minimax scheduling framework is proposed in \cite{YanShe:J22a} which can ensure the safety and achieve a robust performance. However, since the minimax framework is developed on a fully non-cooperative system model, it is inefficient in mixed cooperative and non-cooperative systems as it fails to exploit the benefit of cooperation. Specifically, the cooperation between vehicles can help improve the efficiency from the following two aspects.

First, the passing order of cooperative vehicles can be optimized in advance. However, in most existing methods of optimizing the passing order, a centralized controller with strong communication and computation abilities is often required \cite{MedCreLefWou:J20, ZhaLiuNgo:J21}, and thus these methods are costly to implement extensively. A decentralized order optimization method is proposed in \cite{YanShe:C20b}, while it only works for two-road intersections and cannot be extended to more general scenarios.

Second, the trajectory planning policy can be further modified. Note that the policies against non-cooperative vehicles are generally more conservative to ensure the robustness. When facing cooperative vehicles, more efficient trajectories can be generated without increasing the risk of collision. In \cite{MedWouNij:J18}, a trajectory planning framework for cooperative vehicles is proposed; however, it is not applicable for mixed cooperative and non-cooperative systems.

In this letter, we establish a two-stage policy for the systems with mixed cooperative and non-cooperative vehicles to exploit the benefit of cooperation by the two intuitions above. First, a long-horizon self-organization policy is developed to determine the passing order of cooperative vehicles through distributed optimizations and trajectory adjustments. Then, a short-horizon trajectory planning policy is proposed to improve the overall efficiency of cooperative vehicles while ensuring the robustness against non-cooperative vehicles. Moreover, theoretical and numerical results validating the safety and efficiency of proposed policy are also provided.


\section{PROBLEM FORMULATION}\label{sec:mod}
\subsection{System Model}
Consider the system with $R$ roads indexed by $1,\ldots, R$ crossing over an intersection area, as shown in Fig. \ref{fig:1}. On each road $r$, let $(0,d_r)$ be the range of the intersection area. Let $\mathcal{S}_r$ be the set of vehicles on the road $r$, and let $\mathcal{S}\triangleq\bigcup_{r=1}^R \mathcal{S}_r$ be the set of all vehicles. We assume that all vehicles always move along the corresponding road, and we only focus on their longitudinal movement. Both cooperative and non-cooperative vehicles exist in the system, and let $\mathcal{S}^{\text{coop}}$ and $\mathcal{S}^{\text{non-coop}}$ be the corresponding vehicle sets. All cooperative vehicles follow the same policy, which we aim to design in this letter. In contrast, the policy of any non-cooperative vehicle is unknown and uncontrollable.

For each vehicle $s_i\in\mathcal{S}$, the length, front position, velocity and acceleration are denoted by $l_i$, $p_i(t)$, $v_i(t)$ and $a_i(t)$, respectively, where $t$ represents the time. The vector $\langle p_i(t),v_i(t)\rangle$ is also denoted as the state of $s_i$ at time $t$. The velocity and acceleration of the vehicles are constrained by
\begin{equation}
  v_i(t)\in[0,v_{i,\mathrm{M}}],\; a_i(t)\in[-a_{i,\mathrm{m}},a_{i,\mathrm{M}}],\; \forall i,t.
\end{equation}
We assume that $l_i$, $v_{i,\mathrm{M}}$, $a_{i,\mathrm{m}}$ and $a_{i,\mathrm{M}}$ are equal for all cooperative vehicles, and we drop the subscript $i$ in this case. For each vehicle $s_i$, the initial position $p_i(0)\ll 0$, and the initial velocity $v_i(0)=v_{i,\mathrm{M}}$. In order to ensure the safety of a vehicle $s_i\in\mathcal{S}_r$, the following conditions need to be satisfied:
\begin{itemize}
  \item For any other vehicle $s_j\in\mathcal{S}_r$ on the same road, any position $x\in(0,d_r)$ within the intersection area should not be simultaneously occupied by $s_i$ and $s_j$, i.e.,\footnote{This condition implies that we do not allow overtaking in the intersection area. Nevertheless, overtaking can occur outside the intersection area.}
  \begin{equation}
    \{t\,|\,p_i(t)\in(x,x+l_i),p_j(t)\in(x,x+l_j)\}=\varnothing.\label{eq:S1}
  \end{equation}
  \item For any vehicle $s_j\in\mathcal{S}_{r'}$ with $r'\ne r$, the intersection area should not contain $s_i$ and $s_j$ simultaneously, i.e.,
  \begin{equation}
    \{t\,|\,p_i(t)\in(0,l_i+d_r),p_j(t)\in(0,l_j+d_{r'})\}=\varnothing.\label{eq:S2}
  \end{equation}
\end{itemize}

Furthermore, for each cooperative vehicle $s_i\in\mathcal{S}_r$, we make the following assumptions.
\begin{itemize}
  \item $s_i$ knows the existence of all vehicles.
  \item $s_i$ can actively request any desired information from other cooperative vehicles, and the communication delay is smaller than $\tau$.
  \item $s_i$ can passively receive the state information on non-cooperative vehicles, and the delay is not guaranteed.
  \item $s_i$ knows the indexes of the cooperative vehicles right in front of it and right behind it on the same road, denoted by $\mathfrak{F}(s_i)$ and $\mathfrak{B}(s_i)$, respectively;\footnote{For the first (or last) cooperative vehicle on a road, $\mathfrak{F}(s_i)$ (or $\mathfrak{B}(s_i)$) can be defined as a virtual vehicle with position $+\infty$ (or $-\infty$).} and also knows the index of the vehicle 
  \begin{equation}
    \mathfrak{C}_{r'}(s_i)\triangleq\mathop{\arg\min}_{s_j\in\mathcal{S}_{r'}\cap\mathcal{S}^{\text{coop}}} |p_j(0)-p_i(0)|\label{eq:Cr}
  \end{equation}
  for any $r'\ne r$.
\end{itemize}

Our goal is to ensure the safety and improve the overall efficiency for all cooperative vehicles. Specifically, we use the \emph{relative scheduling cost} to evaluate the efficiency of a cooperative vehicle $s_i$, which is defined as follows.
\begin{definition}
The relative scheduling cost of the cooperative vehicle $s_i$ during the time interval $[0,t_0]$ is defined by
\begin{equation}
  C_i(t_0)\triangleq p_i(0)+v_{\mathrm{M}}t_0 -\Big[p_i(t_0)-\frac{(v_{\mathrm{M}}-v_i(t_0))^2}{2a_{\mathrm{M}}}\Big].
\end{equation}
\end{definition}

\vspace{0.5em}
\begin{remark}
Intuitively, if there is no intersection area, the ideal trajectory of $s_i$ is to first accelerate to $v_{\mathrm{M}}$ with the acceleration $a_{\mathrm{M}}$ and then remain the constant velocity. Let $\sigma_0$ be the ideal trajectory of $s_i$ at time $0$, and $\sigma_{t_0}$ be the ideal trajectory of $s_i$ at time $t_0$. Then $C_i(t_0)$ is the relative distance between $\sigma_0$ and $\sigma_{t_0}$ at a sufficiently large time. Therefore, the relative scheduling cost characterizes the influence of the trajectory adjustment during $[0,t_0]$ caused by scheduling.
\end{remark}

%
%
%
%

\subsection{Overview of the Letter}
In our previous work \cite{YanShe:J22a}, we considered the special case with only one cooperative vehicle (i.e., $\mathcal{S}^{\text{coop}}=\{s_0\}$), and formulated a minimax policy for $s_0$ to optimize the worst-case performance against the trajectory uncertainty of the non-cooperative vehicles. Although the minimax policy is applicable in the general case with multiple cooperative vehicles, it is inefficient since it fails to make use of the cooperation among the vehicles in $\mathcal{S}^{\text{coop}}$. Specifically, the following two points can be considered which can significantly improve the overall efficiency.
\begin{itemize}
  \item The passing order of the cooperative vehicles can be determined in advance. Furthermore, it can also be optimized for a better performance. Intuitively, the overall efficiency will be improved if vehicles from the same road pass the intersection area continuously.
  \item The trajectory uncertainty of the cooperative vehicles can be reduced through the policy design. Therefore, many ``worst cases'' due to the uncertainty of non-cooperative vehicles will never occur, and the minimax optimization becomes too conservative.
\end{itemize}
The policy proposed in this letter will address both these two points. Specifically, for the order determination problem, we focus on the ``long-horizon'' stage where the vehicles are far from the intersection area. In contrast, the uncertainty reduction problem is considered in the ``short-horizon'' stage.

\section{SCHEDULING POLICY}
In this section, we separately propose the scheduling policy in the long-horizon and the short-horizon stages. Specifically, we assume that the position ranges of the two stages are $[-A-B,-B]$ and $[-B,0]$, respectively.

\subsection{Long-Horizon Self-Organization}\label{subsec:long}
First, we develop a distributed self-organization policy for the cooperative vehicles within the position range $[-A-B,-B]$. In this stage, we only consider the cooperative vehicles, and the aim of this stage includes the following two aspects.
\begin{itemize}
  \item Determine the passing order of the cooperative vehicles. Without loss of generality, for the vehicles on the same road, the initial order is maintained.
  \item Adjust the relative positions of the cooperative vehicles based on the passing order.
\end{itemize}
\begin{remark}
Note that once the passing order is fixed, some cooperative vehicles have to adjust their trajectories in order to avoid collisions. In the proposed policy, this ineliminable adjustment is performed guadually from the long-horizon stage, so that the amount of adjustment in the short-horizon stage will be reduced, and the complete trajectory of cooperative vehicles will be smoother.
\end{remark}

Instead of separately considering the two aspects, the proposed self-organization policy simultaneously addresses both of them. Specifically, each cooperative vehicle independently collects information, formulates optimization problems, and adjusts its own trajectory. The passing order is then determined according to the time when they leave the position range $[-A-B,-B]$.

Before presenting the policy, we first introduce some basic settings. Let $\{t'_k\,|\,k\ge0\}$ be a set of times shared by all cooperative vehicles,\footnote{Considering the non-synchronization, an error in the scale of millisecond is allowed, which has little effect on the states of vehicles.} and each vehicle determines its trajectory in $[t'_k,t'_{k+1}]$ during the previous interval $[t'_{k-1},t'_k]$. We also request that the instantaneous velocity of any cooperative vehicle at any $t'_k$ is $v_{\mathrm{M}}$. Thus, determining $p_i(t'_{k+1})$ is sufficient for $s_i$ during the interval $[t'_{k-1},t'_k]$. Furthermore, we assume that the gaps between any adjacent $t'_k$ and $t'_{k+1}$ are sufficiently large, so that during $[t'_k,t'_{k+1}]$, the average velocity of any cooperative vehicle can arbitrarily choose from the interval $[v_{\mathrm{M}}-v_{\mathrm{R}},v_{\mathrm{M}}]$, and any cooperative vehicle can communicate with at least $2R+W$ cooperative vehicles, where $0<v_{\mathrm{R}}\ll v_{\mathrm{M}}$ represents the maximum velocity adjustment and $W\ge 1$ represents the scale of the vehicle set that will be used in order optimization.

\SetInd{0.25em}{0.5em}
\begin{algorithm}[!t]
  \caption{Long-horizon self-organization policy for $s_i$ during $[t'_{k-1},t'_k]$}\label{alg:long}
  \KwOut{$p_i(t'_{k+1})$}

  Initialize $\mathcal{S}_i^k=\{s_i\}$, $P_i^k(s_i)=p_i(t'_k)$.

  \While{$|\mathcal{S}_i^k|<W+2R$}{
      \uIf{there exists some road with no vehicles in $\mathcal{S}_i^k$}{
          Add a cooperative vehicle from this road to $\mathcal{S}_i^k$.
      }
      \uElseIf{there exists some road ${r'}$ on which the last vehicle (based on the initial order) in $\mathcal{S}_i^k$, denoted by $s_j$, satisfies $P_i^k(s_j)-d_{r'}/2>P_i^k(s_i)-d_{r}/2$}{
          Add the vehicle $\mathfrak{B}(s_j)$ to $\mathcal{S}_i^k$.
      }
      \Else{
          Find the road ${r'}$ on which the first vehicle (based on the initial order) in $\mathcal{S}_i^k$, denoted by $s_j$, has the smallest $P_i^k(s_j)-d_{r'}/2$ among all roads.

          Add the vehicle $\mathfrak{F}(s_j)$ to $\mathcal{S}_i^k$.
      }

      Reset $P_i^k(s_j)=p_j(t'_k)$ for all $s_j\in\mathcal{S}_i^k$, and on each road ${r'}$, successively update
      \begin{equation}
        P_i^k(s_m)\leftarrow \min\big\{P_i^k(s_m),P_i^k(\mathfrak{F}(s_m))-(l+\delta)\big\}
      \end{equation}
      where $s_m$ varies from the second vehicle to the last vehicle (based on the initial order) in $\mathcal{S}_i^k\cap\mathcal{S}_{r'}$.
  }

  On each road ${r'}$, remove all vehicles $s_j$ satisfying $P_i^k(s_j)-d_{r'}/2<P_i^k(s_i)-d_{r}/2$ from $\mathcal{S}_i^k$.

  \If{$|\mathcal{S}_i^k|>W$}{
      Remove $|\mathcal{S}_i^k|-W$ vehicles $s_j$ with the largest $P_i^k(s_j)-d_{r'}/2$ from $\mathcal{S}_i^k$, where ${r'}$ is the road of $s_j$.
  }

  Obtain $p^*_i$ by solving the optimization problem \eqref{eq:opt}.

  Let $p_i(t'_{k+1})=\max\big\{p_i(t'_k)+(v_{\mathrm{M}}-v_{\mathrm{R}})(t'_{k+1}-t'_k),$ $p^*_i+v_{\mathrm{M}}(t'_{k+1}-t'_k)\big\}$.
\end{algorithm}

The long-horizon policy for a general cooperative vehicle $s_i\in\mathcal{S}_r$ during $[t'_{k-1},t'_k]$ is shown in Algorithm \ref{alg:long}. The policy can be divided into the following two parts.

\emph{1) Lines 1--13: Information Collection.} The goal of these lines is determining a subset of cooperative vehicles $\mathcal{S}_i^k$ and a function $P_i^k$. First, according to Line 10, $P_i^k$ can be regarded as the equivalent vehicle positions at time $t_k$ after taking the initial order of each road into consideration. Then note that the set $\mathcal{S}_i^k$ is obtained by gradually adding vehicles according to Lines 3--9 followed by a truncation in Lines 11--13. By these steps, $\mathcal{S}_i^k$ roughly contains at most $W$ vehicles $s_j$ with the smallest $P_i^k(s_j)-d_{r'}/2$ among those satisfying $P_i^k(s_j)-d_{r'}/2\ge P_i^k(s_i)-d_{r}/2$, where $r'$ represents the road of $s_j$.



\emph{2) Lines 14--15: Order Optimization.} The goal of these lines is determining $p_i(t'_{k+1})$ by solving the following optimization problem based on the obtained information.
  \begin{align}
      \hspace{-1em}\min_{\{p^*_j\}}\; &\sum_{s_j\in\mathcal{S}_i^k}\big(P_i^k(s_j)-p^*_j\big)\label{eq:opt}\\
      \hspace{-1em}\text{s.t.}\;& p^*_j\le p_j(t_k), \quad\forall s_j\in\mathcal{S}_i^k\notag\\
      &p^*_{j_1}-p^*_{j_2}\ge l+\delta, \notag\\
      &\quad\forall s_{j_1},s_{j_2}\in\mathcal{S}_i^k\cap\mathcal{S}_{r'} \text{ with }p_{j_1}(0)>p_{j_2}(0), \forall r'\notag\\
      &p^*_{j_1}-p^*_{j_2}\ge l+\delta+d_{r'_1} \text{ or } p^*_{j_2}-p^*_{j_1}\ge l+\delta+d_{r'_2}, \notag\\
      &\quad\forall s_{j_1}\in\mathcal{S}_i^k\cap\mathcal{S}_{r'_1}, s_{j_2}\in\mathcal{S}_i^k\cap\mathcal{S}_{r'_2}, r'_1\ne r'_2.\notag
  \end{align}
Intuitively, the objective function in \eqref{eq:opt} represents the total position shift of all vehicles in $\mathcal{S}_i^k$; the first constraint corresponds to the feasible condition;\footnote{More specifically, since the ideal trajectory corresponds to the maximum velocity, vehicles can only decelerate to make adjustments.} the last two constraints correspond to the safety conditions \eqref{eq:S1} and \eqref{eq:S2}. Here, $\delta$ is a positive parameter which accounts for the extra safety distance between vehicles. After obtaining $p^*_i$, the output $p_i(t'_{k+1})$ is determined in Line 15 with considering the minimum average velocity for cooperative vehicles, which is $v_{\mathrm{M}}-v_{\mathrm{R}}$.

Now we provide further analyses on the proposed self-organization policy. Note that for different cooperative vehicles, the optimization problems \eqref{eq:opt} also vary from each other. Nevertheless, in the next theorem, we show that this mismatch will not influence the effectiveness of the proposed policy.

\begin{theorem}\label{thm:1}
Let $\mathcal{S}^{\text{coop}}$ be a finite set, and assume that $p_i(0)$($\forall s_i\in\mathcal{S}^{\text{coop}}$), $d_r$($\forall r$), $v_{\mathrm{M}}(t'_{k+1}-t'_k)$($\forall k$), $v_{\mathrm{R}}(t'_{k+1}-t'_k)$($\forall k$) and $l+\delta$ are linearly independent under integer coefficients. If all cooperative vehicles always follow the proposed self-organization policy, then there exists $t_{\mathrm{c}}<+\infty$ such that (i) all cooperative vehicles maintain the maximum velocity $v_{\mathrm{M}}$ after time $t_{\mathrm{c}}$, and (ii) any two cooperative vehicles $s_{j_1}\in\mathcal{S}_{r'_1}$ and $s_{j_2}\in\mathcal{S}_{r'_2}$ satisfy
\begin{align}
  & p_{j_1}(t_{\mathrm{c}})-p_{j_2}(t_{\mathrm{c}}) \ge l+\delta+d_{r'_1}\cdot \mathds{1}_{r'_1\ne r'_2} \notag\\
  \text{or}\quad & p_{j_2}(t_{\mathrm{c}})-p_{j_1}(t_{\mathrm{c}}) \ge l+\delta+d_{r'_2}\cdot \mathds{1}_{r'_1\ne r'_2}\label{eq:conv}
\end{align}
where $\mathds{1}$ is the indicator function.
\end{theorem}

\begin{proof}
See Appendix \ref{pf:thm1}.
\end{proof}

Intuitively, \eqref{eq:conv} means that the positions of $s_{j_1}$ and $s_{j_2}$ have been sufficiently separated such that these two vehicles will not collide if they maintain the maximum velocity in the future. Therefore, Theorem \ref{thm:1} means that if the proposed policy is followed for a sufficiently long time, all cooperative vehicles will finally complete all the necessary position adjustments and reach a consensus on their passing order.

\begin{remark}
The linear independence assumption in Theorem \ref{thm:1} is introduced in order to avoid the case that the positions of two vehicles are equal, and it can be satisfied in practice by adding a small random number to each parameter.
\end{remark}

\subsection{Short-Horizon Trajectory Planning}\label{subsec:short}
In this subsection, we focus on the trajectory planning of the cooperative vehicles within the short-horizon position range $[-B,0]$. Note that both cooperative and non-cooperative vehicles exist in the system, and we aim to improve the efficiency of cooperative vehicles while ensuring their safety. During this stage, the passing order of cooperative vehicles has been fixed.

The main challenge in this stage is that each vehicle needs to consider the existence of both cooperative and non-cooperative vehicles. Specifically, direct following policies for fully cooperative systems such as \cite{MedWouNij:J18} cannot achieve a robust performance when non-cooperative vehicles exist, while the minimax policy proposed in \cite{YanShe:J22a} is designed for fully non-cooperative systems and fails to take advantage of the cooperation of vehicles. To tackle the challenge, we propose a short-horizon trajectory planning policy by modifying the minimax policy, which can exploit the benefit of cooperation while maintaining the robustness against non-cooperative vehicles.

We focus on the policy of a general cooperative vehicle $s_i\in\mathcal{S}_r$, and let $\{t_{i,k}\,|\,k\ge0\}$ be the set of decision times satisfying $p_i(t_{i,0})=-B$.\footnote{Note that $t_{i,k}$ is different from $t'_k$ used in the long-horizon policy.} Let $s_j\in\mathcal{S}_{r'}$ be the last cooperative vehicle in front of $s_i$ based on the passing order. According to our assumption, $s_i$ can obtain part of the position and velocity information of $s_j$ through communication. We assume that at time $t_{i,k}$, the latest information that $s_i$ has obtained on $s_j$ is $\langle p_j(t_{\mathrm{obs}}),v_j(t_{\mathrm{obs}})\rangle$, where $t_{\mathrm{obs}}\le t_{i,k}$.

\SetInd{0.25em}{0.5em}
\begin{algorithm}[!t]
  \caption{Short-horizon trajectory planning policy for $s_i$ at time $t_{i,k}$}\label{alg:2}
  \KwOut{Trajectory of $s_i$ over $[t_{i,k},t_{i,k+1}]$}

  \uIf{$p_i(t_{i,k})+\frac{v_i(t_{i,k})^2}{2a_{\mathrm{m}}}>0$}{
      Return the trajectory with maximum acceleration.
  }
  \Else{
      Find $\sigma^*$, the trajectory with the smallest cost among all safe trajectories against non-cooperative vehicles based on the current knowledge. Let $\langle p^*,v^*\rangle$ be the state of $s_i$ at time $t_{i,k+1}$ corresponding to $\sigma^*$.

      Find $\mathcal{F}_{i,k}^{\mathrm{fol}}$, which contains all possible states $\langle p_{\mathrm{tar}},v_{\mathrm{tar}}\rangle$ of $s_i$ at time $t_{i,k+1}$ satisfying: (i) $p_{\mathrm{tar}}+\frac{v_{\mathrm{tar}}^2}{2a_{\mathrm{m}}}\le0$, and (ii) for any $t_{\mathrm{D}}\ge t_{i,k+1}$, by fixing the state $\langle p_{\mathrm{tar}},v_{\mathrm{tar}}\rangle$ of $s_i$ at $t_{i,k+1}$ and the state $\langle p_j(t_{\mathrm{obs}}),v_j(t_{\mathrm{obs}})\rangle$ of $s_j$ at $t_{\mathrm{obs}}$, and assuming that both $s_i$ and $s_j$ adopt the minimum acceleration before $t_{\mathrm{D}}$ and adopt the maximum acceleration after $t_{\mathrm{D}}$, it holds that
      \begin{equation}
      p_i(t)\le p_j(t)-l-d_{r'}\cdot \mathds{1}_{r\ne r'},\quad \forall t\ge t_{i,k+1}.
      \end{equation}

      \uIf{$p_j(t_{\mathrm{obs}})+\frac{v_j(t_{\mathrm{obs}})^2}{2a_{\mathrm{m}}}>0$ and $p^*+\frac{v^{*2}}{2a_{\mathrm{m}}}>0$}{
        Return $\sigma^*$.
      }
      \uElseIf{$\mathcal{F}_{i,k}^{\mathrm{fol}}=\varnothing$}{
        Return the trajectory with minimum acceleration.
      }
      \Else{
        Obtain $p_{\mathrm{tar}}$ and $v_{\mathrm{tar}}$ by solving the minimax optimization problem \eqref{eq:opt2}.

        Return an arbitrary trajectory over $[t_{i,k},t_{i,k+1}]$ satisfying $p_i(t_{i,k+1})=p_{\mathrm{tar}}$, $v_i(t_{i,k+1})=v_{\mathrm{tar}}$.
      }
  }
\end{algorithm}

The proposed short-horizon trajectory planning policy for $s_i$ at time $t_{i,k}$ is shown in Algorithm \ref{alg:2}. First, similar to the minimax policy in \cite{YanShe:J22a}, Algorithm \ref{alg:2} is based on the following minimax optimization problem (see Line 11)
\begin{equation}
  \mathop{\arg\min}_{\langle p_{\mathrm{tar}},v_{\mathrm{tar}}\rangle \in\mathcal{F}_{i,k}^{\mathrm{fol}}}\; \max_{\mathcal{K}_{k+1}\in\mathscr{J}_{i,k}}\; V(p_{\mathrm{tar}},v_{\mathrm{tar}};\mathcal{K}_{k+1})\label{eq:opt2}
\end{equation}
which aims to achieve a robust performance against the trajectory uncertainty of non-cooperative vehicles. Specifically, the objective function $V(p,v;\mathcal{K})$ provides an evaluation of the state $\langle p,v\rangle$ under the situation information $\mathcal{K}$; the inner optimization aims to find the worst case $\mathcal{K}_{k+1}$ based on the current knowledge $\mathscr{J}_{i,k}$; and the outer minimization aims to find the target state $\langle p_{\mathrm{tar}},v_{\mathrm{tar}}\rangle$ resulting in the best worst-case performance. Detailed definitions of these variables can be found in \cite{YanShe:J22a}.

Despite the similarities above, the following two differences can also be noted, which reflect how the proposed policy deals with the cooperative vehicles.
\begin{itemize}
  \item Different from \cite{YanShe:J22a}, the minimax optimization \eqref{eq:opt2} is not always adopted in Algorithm \ref{alg:2}. Specifically, Lines 1--2 show that $s_i$ will directly adopt the maximum acceleration without solving the optimization problem after its state enters the region
      \begin{equation}
          \mathcal{C}\triangleq\bigg\{\langle p,v\rangle\,\bigg|\,p+\frac{v^2}{2a_{\mathrm{m}}}>0\bigg\}.
      \end{equation}
      Note that the cooperative vehicles whose states lie in $\mathcal{C}$ will certainly enter the intersection area in finite time. Therefore, for these vehicles, further adjustments can only make a slight difference on the time period of their intersection occupation. Contrarily, withdrawing the opportunity of these further adjustments provides the cooperative vehicles with clearer information about the future trajectories of others, which is more beneficial to improve the overall performance.
  \item The feasible space of $\langle p_{\mathrm{tar}},v_{\mathrm{tar}}\rangle$ in the outer minimization of \eqref{eq:opt2} is also different from that in \cite{YanShe:J22a}. Specifically, the two conditions in Line 5 are introduced to reduce the original feasible space to $\mathcal{F}_{i,k}^{\mathrm{fol}}$. These constraints aim to ensure that $s_i$ always maintains a sufficient distance with the last cooperative vehicle $s_j$, so that it can safely pass the intersection area after $s_j$.
\end{itemize}


In the following, we provide two theoretical results on the proposed short-horizon policy.
\begin{theorem}\label{thm:2}
Assume that the length of the short-horizon position range $[-B,0]$ satisfies
\begin{equation}\label{eq:thm2}
B\ge\frac{v_{\mathrm{M}}^2}{a_{\mathrm{m}}}+\frac{v_{\mathrm{M}}^2}{2a_{\mathrm{M}}} +l+\max\big\{d_{r'}\,\big|\,1\le r'\le R\big\}.
\end{equation}
If all cooperative vehicles follow Algorithm \ref{alg:2} in the range $[-B,0]$ and adopt the maximum acceleration after passing the position $0$, then all cooperative vehicles can safely pass the intersection area.
\end{theorem}

\begin{proof}
See Appendix \ref{pf:thm2}.
\end{proof}

\begin{theorem}\label{thm:3}
Assume that there are no non-cooperative vehicles in the system, and all cooperative vehicles follow Algorithm \ref{alg:2} in the range $[-B,0]$. Let $\mu$ be an upper bound of $t_{i,k+1}-t_{i,k}$ for all $s_i\in\mathcal{S}^{\text{coop}}$ and $k\ge0$, and define $\delta^*\triangleq v_{\mathrm{M}}(\tau+\mu)(1+\frac{a_{\mathrm{m}}}{a_{\mathrm{M}}})$.\footnote{Recall that $\tau$ is the maximum communication delay among cooperative vehicles defined in Section \ref{sec:mod}.} If $v_i(t_{i,0})=v_{\mathrm{M}}$ holds for any $s_i\in\mathcal{S}^{\text{coop}}$ and
\begin{align}
  & v_{\mathrm{M}}(t_{j_2,0}-t_{j_1,0}) \ge l+\delta^*+d_{r'_1}\cdot \mathds{1}_{r'_1\ne r'_2} \notag\\
  \text{or}\quad & v_{\mathrm{M}}(t_{j_1,0}-t_{j_2,0}) \ge l+\delta^*+d_{r'_2}\cdot \mathds{1}_{r'_1\ne r'_2}\label{eq:cond}
\end{align}
holds for any $s_{j_1}\in\mathcal{S}^{\text{coop}}\cap\mathcal{S}_{r'_1}$ and $s_{j_2}\in\mathcal{S}^{\text{coop}}\cap\mathcal{S}_{r'_2}$, then all cooperative vehicles will maintain the maximum velocity $v_{\mathrm{M}}$ in the position range $[-B,0]$.
\end{theorem}

\begin{proof}
See Appendix \ref{pf:thm3}.
\end{proof}

Intuitively, Theorem \ref{thm:2} claims that if the length of the position range $[-B,0]$ is not too short, then Algorithm \ref{alg:2} can ensure the safety of cooperative vehicles. On the other hand, Theorem \ref{thm:3} verifies the efficiency of Algorithm \ref{alg:2}. Specifically, by noting the similarity between \eqref{eq:cond} and \eqref{eq:conv}, the conditions in Theorem \ref{thm:3} mean that all vehicles are cooperative, and they have completed the long-horizon position adjustments where the safety distance parameter $\delta$ is not smaller than $\delta^*$. Under these conditions, Theorem \ref{thm:3} claims that vehicles will not further decelerate in the short-horizon stage. Note that this property is not true if the minimax policy in \cite{YanShe:J22a} is adopted. Therefore, the proposed short-horizon policy reduces unnecessary loss of efficiency by making full use of the cooperation relation.

\section{NUMERICAL RESULTS}\label{sec:simu}
In this section, we provide numerical results on the proposed long-horizon and short-horizon policies. In simulations, we set $R=3$, $d_r=5$ ($\forall r$), $l_i=5$, $v_{i,\mathrm{M}}=20$, $a_{i,\mathrm{m}}=4$, and $a_{i,\mathrm{M}}=3$ ($\forall s_i$). Furthermore, let $v_{\mathrm{R}}=1$, $\tau=0.1$, $\mu=0.1$, $\delta=\delta^*\approx 9.33$. The gaps between adjacent $t'_k$ are set to $2$, and for each cooperative vehicle $s_i$, the gaps between adjacent $t_{i,k}$ are set to $0.1$.\footnote{By adding SI units (i.e., $\mathrm{m}$, $\mathrm{m/s}$, etc.), our choice of parameters is close to the real-world scenarios. Furthermore, we add a small random number to each parameter to meet the conditions of Theorem \ref{thm:1}.}

We assume that the time interval between two cooperative vehicles passing the position $-A-B$ follows independent exponential distribution with expectation $1/\lambda^{\text{coop}}$, and the time interval between two non-cooperative vehicles passing the position $-A-B$ is a constant $1/\lambda^{\text{non-coop}}$. Thus, $\lambda^{\text{coop}}$ and $\lambda^{\text{non-coop}}$ are the traffic rates of cooperative and non-cooperative vehicles, respectively. Furthermore, each new vehicle randomly occurs on the three roads with equal probability, and its initial velocity is $v_{\mathrm{M}}$.

\begin{figure} [!t]
  \centering%
  \includegraphics[width=0.48\textwidth]{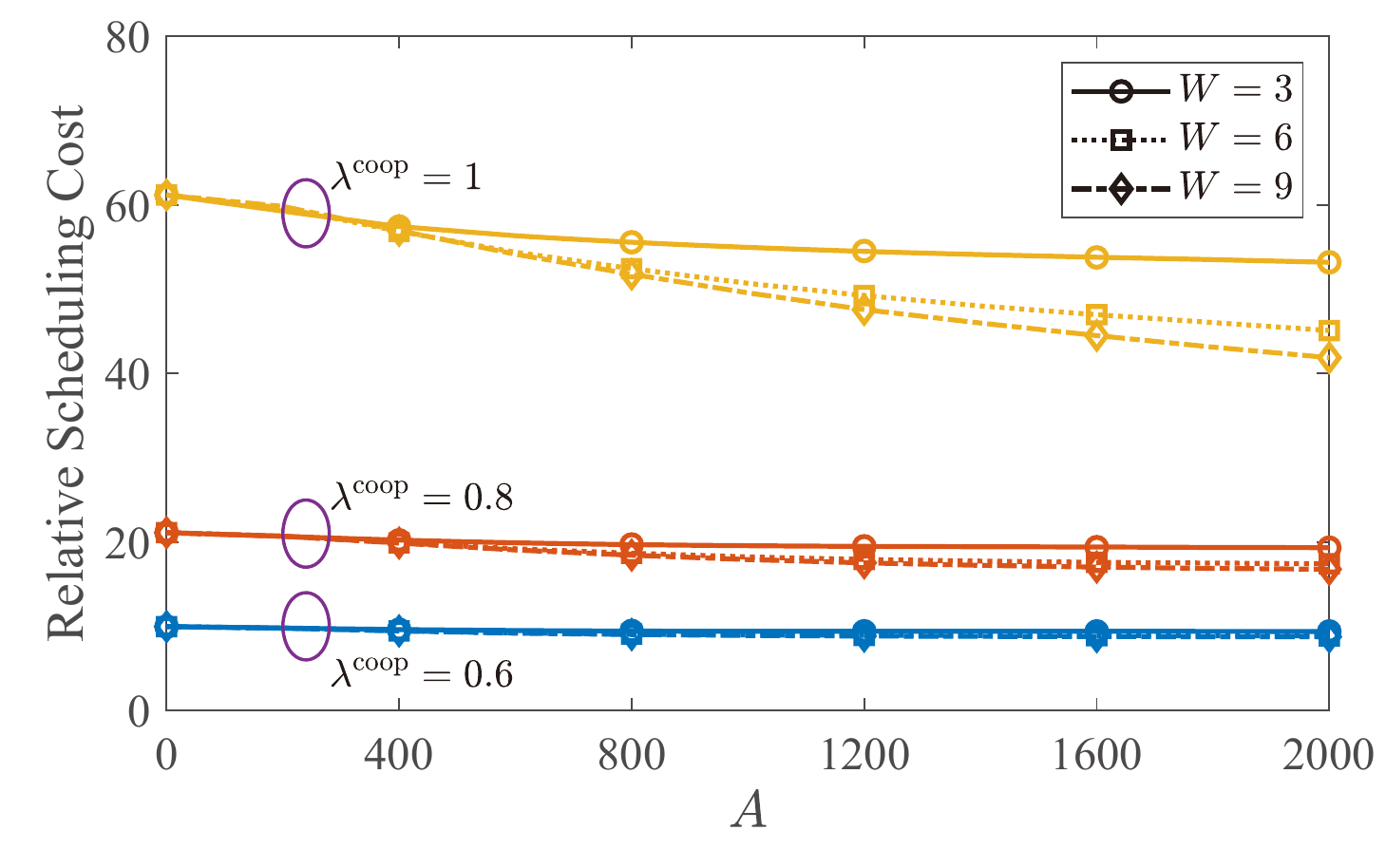}
  \caption{The average relative scheduling cost of cooperative vehicles under different $A$, $W$ and $\lambda^{\text{coop}}$. Parameters $\lambda^{\text{non-coop}}=0$ and $B=200$. For each data point, $1000$ simulations are performed, and the time span of input vehicles is $500$ in each simulation.}
  \label{fig:3}
\end{figure}

First, we let $\lambda^{\text{non-coop}}=0$, in which case the long-horizon policy plays a key role on the overall performance according to Theorem \ref{thm:3}. The average scheduling cost for different $A$, $W$ and $\lambda^{\text{coop}}$ is shown in Fig. \ref{fig:3}. Specifically, note that when $A=0$, the long-horizon policy is equivalent to following the initial order of cooperative vehicles. 
First, compared with the initial order, the proposed long-horizon policy makes little difference when $\lambda^{\text{coop}}$ is low (e.g., $\lambda^{\text{coop}}=0.6$), while it substantially improves the efficiency when $\lambda^{\text{coop}}$ is high (e.g., $\lambda^{\text{coop}}=1$). Intuitively, the benefit results from gathering the cooperative vehicles from the same road that are initially close to each other and letting them continuously pass the intersection area. In dense traffic, vehicles are closer and thus optimizing the passing order has a larger potential.
Second, for a large $\lambda^{\text{coop}}$, a better performance can be achieved by enlarging $W$ and $A$. Specifically, $W$ influences the achievable efficiency after convergence, while $A$ determines the progress of the convergence process. Note that $A$ is often a predetermined parameter in practice, and $W$ can be freely chosen. Therefore, a small $W$ is often sufficient when $A$ is small, while for a larger $A$, $W$ should also be chosen larger.\footnote{Recall that $W$ also influences the computational complexity of the optimization problem \eqref{eq:opt}. Thus, it should not be chosen too large.}


\begin{figure} [!t]
  \centering%
  \includegraphics[width=0.48\textwidth]{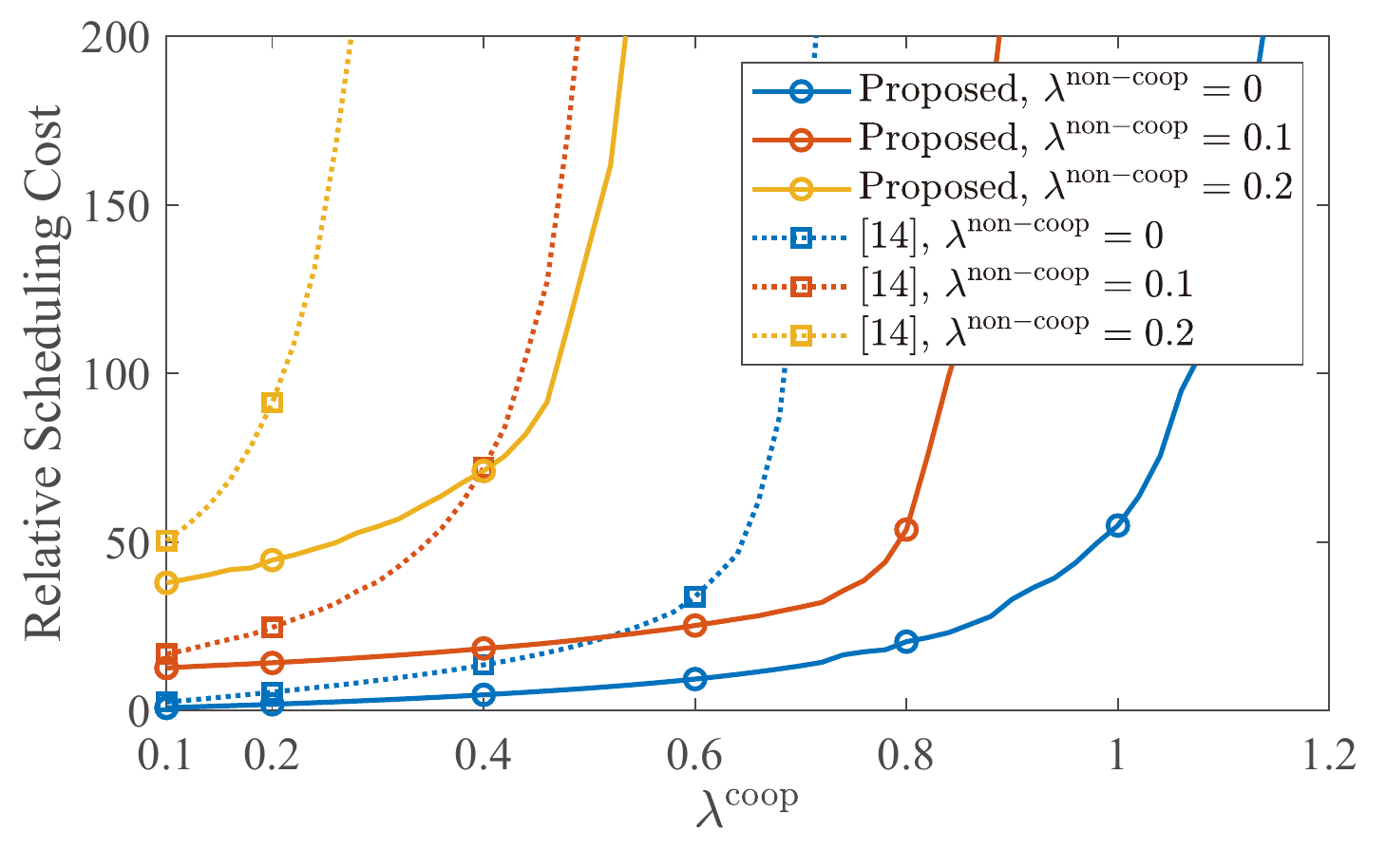}
  \caption{The average relative scheduling cost of cooperative vehicles under different $\lambda^{\text{coop}}$ and $\lambda^{\text{non-coop}}$, where both the minimax policy in \cite{YanShe:J22a} and Algorithm \ref{alg:2} are adopted for the short-horizon stage. Parameters $W=6$, $A=600$ and $B=200$. For each data point, $1000$ simulations are performed, and the time span of input vehicles is $500$ in each simulation.}
  \label{fig:4}
\end{figure}

Then we fix the long-horizon parameters $A=600$ and $W=6$, and focus on the effect of different short-horizon policies. Fig. \ref{fig:4} provides the relative scheduling cost for different $\lambda^{\text{coop}}$ and $\lambda^{\text{non-coop}}$, and both the minimax policy in \cite{YanShe:J22a} and Algorithm \ref{alg:2} are adopted as the short-horizon policy. According to Fig. \ref{fig:4}, for any choice of $\lambda^{\text{non-coop}}$, both a better average efficiency and a larger throughput for cooperative vehicles can be obtained by following the proposed Algorithm \ref{alg:2}. This observation implies that Algorithm \ref{alg:2} exploits the benefits of cooperation and improves the overall performance.

\section{CONCLUSION}\label{sec:conc}
This letter established a distributed two-stage policy for the cooperative vehicles passing an intersection, where both cooperative and non-cooperative vehicles exist. Specifically, the long-horizon stage aims to optimize the passing order of the cooperative vehicles by self organization, where each cooperative vehicle independently formulates optimization problems and adjusts the trajectory; the short-horizon stage aims to generate safe and efficient trajectories when cooperative vehicles approach the intersection area. Numerical experiments showed that both the long-horizon and the short-horizon policies significantly improve the efficiency of cooperative vehicles during the scheduling.

\appendix

\subsection{Proof of Theorem \ref{thm:1}}\label{pf:thm1}
We first introduce the following lemma.
\begin{lemma}\label{lemma:1}
Consider the optimization problem \eqref{eq:opt}. Assume that
\begin{align*}
  &P_i^k(s_{j_1})-P_i^k(s_{j_2})\ge l+\delta, \\
  &\quad\forall s_{j_1},s_{j_2}\in\mathcal{S}_i^k\cap\mathcal{S}_{r'}\setminus\{s_i\} \text{ with }p_{j_1}(0)>p_{j_2}(0), \forall r'\\
  &P_i^k(s_{j_1})-P_i^k(s_{j_2})\ge l+\delta+d_{r'_1} \\
  &\quad\text{or}\quad P_i^k(s_{j_2})-P_i^k(s_{j_1})\ge l+\delta+d_{r'_2}, \\
  &\quad\forall s_{j_1}\in\mathcal{S}_i^k\cap\mathcal{S}_{r'_1}\setminus\{s_i\}, s_{j_2}\in\mathcal{S}_i^k\cap\mathcal{S}_{r'_2}\setminus\{s_i\}, r'_1\ne r'_2.
\end{align*}
Let $s_m\in\mathcal{S}_{r^*}$ be the vehicle in $\mathcal{S}_i^k\setminus\{s_i\}$ with the smallest $P_i^k$, and assume that $P_i^k(s_m)-d_ {r^*}/2 \ne P_i^k(s_i)-d_r/2$. Then the solution of \eqref{eq:opt} satisfies
\begin{equation}
  p_i^*=\big\{P_i^k(s_m)-(l+\delta+d_{r^*}\mathds{1}_{r^*\ne r}),P_i(s_i)\big\}.\label{eq:lemma1}
\end{equation}
\end{lemma}
\vspace{0.5em}

\begin{proof}
Note that the solution $\{p^*_j\}$ of \eqref{eq:opt} also gives an order for all vehicles in $\mathcal{S}_i^k$. Based on this order, we divide $\mathcal{S}_i^k$ into the following three parts:
\begin{itemize}
  \item the ego-vehicle $s_i$;
  \item a subset $\mathcal{S}_{\text{later}}$ containing the vehicles that pass the intersection after $s_i$;
  \item a subset $\mathcal{S}_{\text{earlier}}$ containing the vehicles that pass the intersection before $s_i$.
\end{itemize}
We only need to prove that $\mathcal{S}_{\text{later}}=\varnothing$.

Now we prove by contradiction. If $\mathcal{S}_{\text{later}}$ is non-empty, we first note that vehicles in $\mathcal{S}_{\text{later}}$ must be on different roads from $s_i$, since the initial order of the vehicles from the same road is preserved. In the following, we will define another $\{\hat{p}^*_j\}$ and show that it has a better performance than $\{p^*_j\}$.

First, we consider the case that $s_m\in \mathcal{S}_{\text{earlier}}$. We can define $\hat{p}^*_j=P_i^k(s_j)$ for any $s_j\in\mathcal{S}_{\text{later}}$ and define $\hat{p}^*_j=p^*_j$ for other $s_j$. We can easily check that $\{\hat{p}^*_j\}$ satisfies all the constraints and results in a smaller objective function.

Then, we consider the case that $s_m\in \mathcal{S}_{\text{later}}$, and thus $s_m$ is on a different road from $s_i$. Now we define $\{\hat{p}^*_j\}$ as follows.
\begin{itemize}
  \item For $s_i$, $\hat{p}^*_i$ is defined by the right-hand side of \eqref{eq:lemma1}.
  \item For vehicles $s_j\in\mathcal{S}_{\text{later}}$, take $\hat{p}^*_j=P_i^k(s_j)$.
  \item For vehicles $s_j\in\mathcal{S}_{\text{earlier}}$, take $\hat{p}^*_j=p^*_j$.
\end{itemize}
It is easy to check that $\{\hat{p}^*_j\}$ satisfies the constraints. Then we analyze the objective function.
\begin{itemize}
  \item If $P_i^k(s_m)-(l+\delta+d_{r^*})\ge P_i^k(s_i)$, then we have $\hat{p}^*_i=P_i^k(s_i)\le p^*_i$. Thus, $\hat{p}^*_j\le p^*_j$ holds for all $s_j$, and the inequality strictly holds for the vehicles in $\mathcal{S}_{\text{later}}$. Therefore, the objective function corresponding to $\hat{p}^*_i$ is smaller.
  \item If $P_i^k(s_m)-(l+\delta+d_{r^*})\le P_i^k(s_i)$, then
  \begin{align}
      \hat{p}^*_i-p^*_i&\ge \hat{p}^*_i-P_i^k(s_i)\notag\\
      &=\big[P_i^k(s_m)-(l+\delta+d_{r^*})\big]-P_i^k(s_i)\notag\\
      &> (d_{r^*}-d_r)/2-(l+\delta+d_{r^*})\notag\\
      &=-(l+\delta)-(d_{r^*}+d_r)/2\label{eq:pflemma1}
      \end{align}
      \begin{align}
      \hat{p}^*_m-p^*_m&\ge \hat{p}^*_m-\big[P_i^k(s_i)-(l+\delta+d_r)\big]\notag\\
      &=P_i^k(s_m)-\big[P_i^k(s_i)-(l+\delta+d_r)\big]\notag\\
      &> (d_{r^*}-d_r)/2+(l+\delta+d_r)\notag\\
      &=(l+\delta)+(d_{r^*}+d_r)/2.\label{eq:pflemma2}
  \end{align}
  The last inequalities in \eqref{eq:pflemma1} and \eqref{eq:pflemma2} come from the truncation step when defining the set $\mathcal{S}_i^k$ and the assumption $P_i^k(s_m)-d_ {r^*}/2 \ne P_i^k(s_i)-d_r/2$ in the lemma. Thus, the total cost of $s_i$ and $s_m$ decreases. Noting that the cost of other vehicles are non-increasing, we conclude that the objective function corresponding to $\hat{p}^*_i$ is smaller.
\end{itemize}

The discussions above illustrates the contradiction. Therefore, we have $\mathcal{S}_{\text{later}}=\varnothing$.
\end{proof}

Now we provide the proof of Theorem \ref{thm:1}.

\begin{proof}
Assume that there are $N^{\text{coop}}$ cooperative vehicles in the system. For any $1\le n\le N^{\text{coop}}$ and any decision time $t'_k$, define $\mathcal{W}_{k,n}$ to be the subset of $\mathcal{S}^{\text{coop}}$ which includes the $n$ vehicles $s_i$ with the largest $\tilde{p}_i(t'_k)\triangleq p_i(t'_k)-d_r/2$, where $r$ is the road index corresponding to $s_i$. According to the linear independence assumption in the theorem, the values of $\tilde{p}_i(t'_k)$ for different $s_i\in\mathcal{S}^{\text{coop}}$ are unequal.

We will prove the theorem by induction. Specifically, we will show that for any $1\le n\le N^{\text{coop}}$, there exists some $k_n\ge0$ such that
\begin{itemize}
  \item the set $\mathcal{W}_{k,n}$ is the same for any $k\ge k_n$;
  \item at time $t'_{k_n}$, we have (i) any two cooperative vehicles $s_{j_1}\in\mathcal{W}_{k,n}\cap\mathcal{S}_{r'_1}$ and $s_{j_2}\in\mathcal{W}_{k,n}\cap\mathcal{S}_{r'_2}$ satisfy
\begin{align}
  & p_{j_1}(t'_{k_n})-p_{j_2}(t'_{k_n}) \ge l+\delta+d_{r'_1}\cdot \mathds{1}_{r'_1\ne r'_2} \notag\\
  \text{or}\quad & p_{j_2}(t'_{k_n})-p_{j_1}(t'_{k_n}) \ge l+\delta+d_{r'_2}\cdot \mathds{1}_{r'_1\ne r'_2}
\end{align}
and (ii) all vehicles in $\mathcal{W}_{k,n}$ maintain the maximum velocity $v_{\mathrm{M}}$ after time $t'_{k_n}$.
\end{itemize}

For the case $n=1$, we can choose $k_1=0$. This is because for the vehicle $s_i$ with the largest $\tilde{p}_i(t'_0)$, the corresponding vehicle set $\mathcal{S}_i^0=\{s_i\}$ holds. Therefore, it will maintain the maximum velocity $v_{\mathrm{M}}$ during $[t'_0,t'_1]$. Thus we can inductively check that it will always be the vehicle with the largest $\tilde{p}_i(t'_k)$, and hence $\mathcal{W}_{k,1}$ remains constant for $k\ge0$.

Now we consider the general case with $2\le n\le N^{\text{coop}}$ by assuming that the induction for $n-1$ has completed. First we introduce some definitions. Let $s_m\in\mathcal{S}_{r^*}$ be the vehicle in $\mathcal{W}_{k_{n-1},n-1}$ with the smallest $\tilde{p}_m(t'_{k_{n-1}})$. On each road, we take a vehicle which is outside the set $\mathcal{W}_{k_{n-1},n-1}$ and has the largest initial position, and we collect these vehicles and form a set $\mathcal{A}$. For any vehicle $s_i\in\mathcal{A}$, define
\begin{equation}\label{eq:pfthm1_1}
  G_i(t'_k)\triangleq p_m(t'_k)-p_i(t'_k)-(l+\delta+d_{r^*}\mathds{1}_{r\ne r^*})
\end{equation}
where $r$ is the road index corresponding to $s_i$.

For a given $k$, define the vehicle
\begin{equation}\label{eq:pfthm1_2}
s_{i_0(k)}\triangleq\mathop{\arg\max}_{s_i\in\mathcal{A}}\;\tilde{p}_i(t'_k).
\end{equation}
We aim to show that there exists a large enough $k'_n$ such that
\begin{equation}\label{eq:pfthm1_3}
G_{i_0(k'_n)}(t'_{k'_n})\ge0.
\end{equation}
Actually, if \eqref{eq:pfthm1_3} holds, then we can complete the induction through the following steps.
\begin{itemize}
  \item[(i)] According to \eqref{eq:pfthm1_2}, we have $\mathcal{S}_{i_0(k'_n)}^{k'_n}\subseteq \mathcal{W}_{k_{n-1},n-1}\cup\{s_{i_0(k'_n)}\}$. Thus, we can check that the corresponding optimization problem \eqref{eq:opt} satisfies the condition of Lemma \ref{lemma:1}.
  \item[(ii)] According to Lemma \ref{lemma:1}, the vehicle $s_{i_0(k'_n)}$ will maintain the maximum velocity $v_{\mathrm{M}}$ within the time period $[t'_{k'_n},t'_{k'_n+1}]$. Thus, we have $s_{i_0(k'_n+1)}=s_{i_0(k'_n)}$.
  \item[(iii)] By repeating the steps (i) and (ii), we can conclude that $s_{i_0(k'_n)}$ will always maintain the velocity $v_{\mathrm{M}}$ after time $t'_{k'_n}$, and $s_{i_0(k)}=s_{i_0(k'_n)}$ holds for all $k\ge k'_n$.
  \item[(iv)] Since the initial order is preserved for the vehicles from the same road, there exists $k_n\ge k'_n$ such that $\mathcal{W}_{k_n,n}=\mathcal{W}_{k_{n-1},n-1}\cup\{s_{i_0(k'_n)}\}$. Then we can easily check that the induction has been completed.
\end{itemize}

Finally, we show that \eqref{eq:pfthm1_3} holds for some $k'_n$. The following two results can firstly be obtained.
\begin{itemize}
  \item For any $s_i\in\mathcal{A}$, $G_i(t'_k)$ is non-decreasing with respect to $t'_k$.
  \item For each $k\ge k_{n-1}$, if $G_{i_0(k)}(t'_k)<0$, then either $G_{i_0(k)}(t'_{k+1})=0$ holds, or $G_ {i_0(k)}(t'_{k+1})-G_{i_0(k)}(t'_k)$ is not smaller than a positive constant.
\end{itemize}
Specifically, the first result can be directly checked by the definition \eqref{eq:pfthm1_1} and the induction hypothesis; to obtain the second result, we should note that the optimization problem corresponding to $s_{i_0(k'_n)}$ in the time interval $[t'_{k'_n-1},t'_{k'_n}]$ satisfies the conditions in Lemma \ref{lemma:1}, and then utilize the result of Lemma \ref{lemma:1} and Line 15 of Algorithm \ref{alg:long}.

Note that for any $k$, $G_{i_0(k)}(t'_k)$ must be in a finite set $\{G_i(t'_k)\,|\,s_i\in\mathcal{A}\}$. According to the two results above, there cannot exist an infinite number of $k$ such that $G_{i_0(k)}(t'_k)<0$ holds. Therefore, $k'_n$ must exist.
\end{proof}

\subsection{Proof of Theorem \ref{thm:2}}\label{pf:thm2}
\begin{proof}
According to Algorithm \ref{alg:2}, if a cooperative vehicle $s_i$ has not entered the region $\mathcal{C}$ at time $t_k$, then there exist three possibilities which correspond to Lines 6--7 (Case (a)), Lines 8--9 (Case (b)) and Lines 10--12 (Case (c)) in the algorithm, respectively. Note that Case (a) will only occur at the last decision time before $s_i$ enters the region $\mathcal{C}$, and we denote this decision time as $t_{i,k_0}$.

According to the definition of $\sigma^*$ in Line 4, the output trajectory of $s_i$ after $t_{i,k_0}$ must be conflict-free with all non-cooperative vehicles. In the following, we will show that $s_i$ is conflict-free with $s_j$, which is the last cooperative vehicle in front of $s_i$ based on the passing order. Specifically, by defining
\begin{equation}\label{eq:L}
  L_j\triangleq l+d_{r'}\cdot \mathds{1}_{r\ne r'}
\end{equation}
where $r$ and $r'$ are the road index of $s_i$ and $s_j$ respectively, we will prove that
\begin{equation}\label{eq:pfthm2_1}
  \{t\,|\,p_i(t)\in(x,+\infty),p_j(t)\in(x,x+L_j)\}=\varnothing
\end{equation}
holds for any $x>0$. We can easily check that \eqref{eq:pfthm2_1} is a sufficient condition for the safety condition \eqref{eq:S1} or \eqref{eq:S2}. Now we separately consider the following two situations.

(1) First, we assume that Case (c) never occurs in the decision of $s_i$. In other words, before Case (a) occurs, all decision times correspond to Case (b). Then $s_0$ always follows the minimum acceleration before time $t_{i,k_0}$. Now we consider the decision of $s_0$ at $t_{i,k_0}$. Based on the ``if''-sentence in Line 6, $p_j(t_{\mathrm{obs}})+\frac{v_j(t_{\mathrm{obs}})^2}{2a_{\mathrm{m}}}>0$ must hold. Since $t_{i,k_0}\ge t_{\mathrm{obs}}$, we know that $s_j$ has entered the region $\mathcal{C}$ at $t_{i,k_0}$, and thus it will always follow the maximum acceleration after $t_{i,k_0}$. Therefore, for any $t\ge t_{i,k_0}$, we have
\begin{align}
  p_j(t)-p_i(t)&\ge \Big(p_j(t_{i,k_0})-\frac{v_{\mathrm{M}}^2}{2a_{\mathrm{M}}} +v_{\mathrm{M}}(t-t_{i,k_0})\Big)\notag\\
  &\quad\qquad -\big(p_i(t_{i,k_0})+v_{\mathrm{M}}(t-t_{i,k_0})\big)\notag\\
  &\ge p_j(t_{\mathrm{obs}})-p_i(t_{i,k_0})-\frac{v_{\mathrm{M}}^2}{2a_{\mathrm{M}}}\notag\\
  &\ge -\frac{v_{\mathrm{M}}^2}{2a_{\mathrm{m}}} -\Big(-B+\frac{v_{\mathrm{M}}^2}{2a_{\mathrm{m}}}\Big) -\frac{v_{\mathrm{M}}^2}{2a_{\mathrm{M}}}\notag\\
  &\ge l+\max\big\{d_{r'}\,\big|\,1\le r'\le R\big\} \ge L_j.
\end{align}
Hence, \eqref{eq:pfthm2_1} holds.

(2) Second, we assume that Case (c) occurs in the decision of $s_i$, and let $t_{i,k_1}$ be the decision time when Case (c) first occurs. Now we prove by induction that any decision time $t_{i,k}$ with $k_1< k< k_0$ also corresponds to Case (c), i.e., the set $\mathcal{F}_{i,k}^{\mathrm{fol}}$ is non-empty. If the set $\mathcal{F}_{i,k-1}^{\mathrm{fol}}$ is non-empty, then $\langle p_i(t_{i,k}),v_i(t_{i,k})\rangle$ must belong to it by Lines 11--12. Then we consider the trajectory $\sigma_{\mathrm{D}}$ starting from the state $\langle p_i(t_{i,k}),v_i(t_{i,k})\rangle$ and following the minimum acceleration. We can easily check that under this trajectory, $\langle p_i(t_{i,k+1}),v_i(t_{i,k+1})\rangle\in\mathcal{F}_{i,k}^{\mathrm{fol}}$. Thus we complete the induction.

Therefore, the decision time $t_{i,k_0-1}$ also corresponds to Case (c). Since $t_{i,k_0}$ corresponds to Case (a), $s_j$ must have entered the region $\mathcal{C}$ before $t_{i,k_0}$ and will always follow the maximum acceleration after $t_{i,k_0}$. Let $\sigma_{\mathrm{A}}$ be the trajectory of $s_i$ starting from the state $\langle p_i(t_{i,k_0}),v_i(t_{i,k_0})\rangle$ and following the maximum acceleration. Then, noting that $\langle p_i(t_{i,k_0}),v_i(t_{i,k_0})\rangle\in\mathcal{F}_{i,k_0-1}^{\mathrm{fol}}$ and according to the definition of $\mathcal{F}_{i,k_0-1}^{\mathrm{fol}}$ in Line 5, we must have $p_i(t)\le p_j(t)-L_j$ for any $t\ge t_{i,k_0}$ if $s_i$ follows the trajectory $\sigma_{\mathrm{A}}$. Thus, $p_i(t)\le p_j(t)-L_j$ $(\forall t\ge t_{i,k_0})$ also holds for any trajectory of $s_i$ after $t_{i,k_0}$. In other words, \eqref{eq:pfthm2_1} holds.
\end{proof}

\subsection{Proof of Theorem \ref{thm:3}}\label{pf:thm3}
\begin{proof}
We perform induction on the time. Assume that the outputs of all decisions of the cooperative vehicles before time $t$ are the trajectories with the maximum velocity. Then we consider the first short-horizon decision made by the vehicles after $t$, for which the corresponding vehicle is denoted by $s_i$ and the decision time is $t_{i,k}$. According to the induction hypothesis, all vehicles will follow the maximum velocity $v_{\mathrm{M}}$ before $t_{i,k}$.

If $s_i$ is the first vehicle based on the passing order, then it will always follow the maximum velocity, and the induction is completed.

In the following, we assume that $s_i$ is not the first vehicle, and denote $s_j$ to be the last cooperative vehicle in front of $s_i$ based on the passing order. Assume that for $s_i$, the last state observation on $s_j$ is $\langle p_j(t_{\mathrm{obs}}),v_j(t_{\mathrm{obs}})\rangle$, and we must have $t_{\mathrm{obs}}\ge t_{i,k}-\tau$. Now we focus on Algorithm \ref{alg:2}. Since there are no non-cooperative vehicles in the system, the trajectory $\sigma^*$ defined in Line 4 is actually the trajectory with the maximum acceleration. Therefore, we have
\begin{align}
  &p^*+\frac{v^{*2}}{2a_{\mathrm{m}}}\notag\\
  &= \big[p_i(t_{i,k})+v_{\mathrm{M}}(t_{i,k+1}-t_{i,k})\big]+\frac{v_{\mathrm{M}}^2}{2a_{\mathrm{m}}}\notag\\
  &\le \big[p_j(t_{i,k})-L_j-\delta^*\big]+v_{\mathrm{M}}(t_{i,k+1}-t_{i,k}) +\frac{v_{\mathrm{M}}^2}{2a_{\mathrm{m}}}\notag\\
  &= \big[p_j(t_{\mathrm{obs}})+v_{\mathrm{M}}(t_{i,k}-t_{\mathrm{obs}})\big] -L_j-\delta^*\notag\\
  &\quad\qquad +v_{\mathrm{M}}(t_{i,k+1}-t_{i,k})+\frac{v_j(t_{\mathrm{obs}})^2}{2a_{\mathrm{m}}}\notag\\
  &\le \Big[p_j(t_{\mathrm{obs}})+\frac{v_j(t_{\mathrm{obs}})^2}{2a_{\mathrm{m}}}\Big]
  -\big(L_j+\delta^*-v_{\mathrm{M}}\tau-v_{\mathrm{M}}\mu\big)\notag\\
  &\le p_j(t_{\mathrm{obs}})+\frac{v_j(t_{\mathrm{obs}})^2}{2a_{\mathrm{m}}}\label{eq:pfthm3}
\end{align}
where $\langle p^*,v^*\rangle$ is defined in Line 4 and $L_j$ is identically defined as in \eqref{eq:L}.

We first consider the case $p^*+\frac{v^{*2}}{2a_{\mathrm{m}}}>0$. According to \eqref{eq:pfthm3}, we also have $p_j(t_{\mathrm{obs}})+\frac{v_j(t_{\mathrm{obs}})^2}{2a_{\mathrm{m}}}>0$. By Line 6, $s_i$ will return the trajectory $\sigma^*$ and the induction is completed.

Then, we consider the case $p^*+\frac{v^{*2}}{2a_{\mathrm{m}}}\le 0$. To start, we prove that $\mathcal{F}_{i,k}^{\mathrm{fol}}\ne\varnothing$ by showing $\langle p^*,v^*\rangle\in \mathcal{F}_{i,k}^{\mathrm{fol}}$. We arbitrarily fix $t_{\mathrm{D}}\ge t_{i,k+1}$, fix the state $\langle p^*,v^*\rangle$ of $s_i$ at $t_{i,k+1}$, and fix the state $\langle p_j(t_{\mathrm{obs}}),v_j(t_{\mathrm{obs}})\rangle$ of $s_j$ at $t_{\mathrm{obs}}$. Then we assume that both $s_i$ and $s_j$ adopt the minimum acceleration before $t_{\mathrm{D}}$ and adopt the maximum acceleration after $t_{\mathrm{D}}$. Under this setting, we can check that the velocity difference between $s_i$ and $s_j$ is always not larger than $a_{\mathrm{m}}(t_{i,k+1}-t_{\mathrm{obs}})\le a_{\mathrm{m}}(\tau+\mu)$. Therefore, the difference between their displacements before time $t_{\mathrm{D}}$ is not larger than $a_ {\mathrm{m}}(\tau+\mu)\cdot\frac{v_{\mathrm{M}}}{a_{\mathrm{m}}}$, and the difference between their displacements after time $t_{\mathrm{D}}$ is not larger than $a_{\mathrm{m}}(\tau+\mu)\cdot\frac{v_{\mathrm{M}}}{a_{\mathrm{M}}}$. According to \eqref{eq:cond}, $p_i(t)\le p_j(t)-L_j$ always holds. Therefore, according to the definition of $\mathcal{F}_{i,k}^{\mathrm{fol}}$ in Line 5, we have $\langle p^*,v^*\rangle\in \mathcal{F}_{i,k}^{\mathrm{fol}}$.

Since $\mathcal{F}_{i,k}^{\mathrm{fol}}$ is non-empty, the algorithm goes to Lines 11--12, and the minimax optimization problem will be solved. However, since there are no non-cooperative vehicles, by noting that $\langle p^*,v^*\rangle\in \mathcal{F}_{i,k}^{\mathrm{fol}}$, we can easily check that $\langle p^*,v^*\rangle$ will be the result of the optimization and the policy will return the trajectory with the maximum acceleration. So far, we have completed the induction in all scenarios.
\end{proof}





\end{document}